\documentclass[11pt]{article}
\usepackage{fullpage}
\usepackage{amssymb}
\usepackage{amsmath,amsthm}
\usepackage{multirow}
\usepackage{color}
\usepackage{xspace}

\usepackage{ifpdf}
\ifpdf    
\usepackage{hyperref}
\else    
\usepackage[hypertex]{hyperref}
\fi

\newcommand{\remove}[1]{}

\newtheorem{theorem}{Theorem}[section]

\newtheorem{lemma}[theorem]{Lemma}

\newcommand{\comment}[1]{}
\newcommand{\junk}[1]{}

\def\veps{\ensuremath{\varepsilon}}

\begin{document}
{{
\title{Sparse Suffix Tree Construction with Small Space}
}
\author{
Philip Bille\thanks{Technical University of Denmark, Email: phbi@imm.dtu.dk} \and Inge Li G\o rtz \thanks{Technical University of Denmark, Email: ilg@imm.dtu.dk} \and Tsvi Kopelowitz \thanks{Weizmann Institute of Science, Email: kopelot@gmail.com}\and
Benjamin Sach\thanks{University of Warwick, Email: sach@dcs.warwick.ac.uk} \and Hjalte Wedel Vildh{\o}j\thanks{Technical University of Denmark, Email: hwvi@imm.dtu.dk}
%
}

\date{}

\maketitle

\begin{abstract}

We consider the problem of constructing a sparse suffix tree (or suffix array) for $b$ suffixes of a given text $T$ of size $n$, using only $O(b)$ words of space during construction time. Breaking the naive bound of $\Omega(nb)$ time for this problem has occupied many algorithmic researchers since a different structure, the (evenly spaced) sparse suffix tree, was introduced by K{\"a}rkk{\"a}inen and Ukkonen in 1996. While in the evenly spaced sparse suffix tree the suffixes considered must be evenly spaced in $T$, here there is no constraint on the locations of the suffixes.

We show that the sparse suffix tree can be constructed in $O(n\log^2b)$ time. To achieve this we develop a technique, which may be of independent interest, that allows to efficiently answer $b$ longest common prefix queries on suffixes of $T$, using only $O(b)$ space. We expect that this technique will prove useful in many other applications in which space usage is a concern. Furthermore, additional tradeoffs between the space usage and the construction time are given.
\end{abstract}

\section{Introduction}

In the sparse suffix tree\footnote{All of the results within apply to the sparse suffix array as well.} problem we are given a string $T=t_1\cdots t_n$ of size $n$, and a list of $b$ \emph{interesting} indices of $T$. The goal is to construct the suffix tree for only those $b$ indices, while using little space during the construction process, which will hopefully be $O(b)$ words. Such a construction can be helpful in situations where an extremely large string is saved in read only memory, and we are interested in indexing only a small set of its suffixes, or if the index of all of the text cannot all fit in available memory. Natural examples are indexing a genomic sequence where only part of the locations are of interest for searching for a given gene, or indexing a book where we are only interested in appearances of a pattern which is a beginning of a paragraph, sentence, or word.

A naive algorithm with $O(nb)$ running time can be easily produced by inserting each suffix one at a time into the suffix tree. However, breaking the naive bound has been a problem that has baffled many algorithmic researchers since a similar flavored problem was introduced by K{\"a}rkk{\"a}inen and Ukkonen in\cite{KU96}. The authors there introduced the sparse suffix tree, and showed an efficient construction for the evenly spaced sparse suffix tree, which is a suffix tree for every $k^{\text{th}}$ suffix in the text. In addition, they discussed how to search for a pattern in a sparse suffix tree, and those ideas were later improved by Kolpakov et. al. in~\cite{KKS11}. However, the question of constructing a sparse suffix tree with no restriction on the sparseness while breaking the naive bound remained open. It should be noted that an efficient solution for a suffix tree on words was already introduced by Andersson et.al in~\cite{ALS99}, and later extended to suffix arrays on words by Ferragina and Fischer in~\cite{FF07}, but their model is restrictive as it assumes that there is a delimiter after each word. In the sparse suffix tree there is no such assumption, and hence it is more general.

\paragraph{Results}
We are the first to break the naive $O(nb)$ algorithm for general sparse suffix trees, by showing how to construct a sparse suffix tree in $O(n\log^2b)$ time, using only $O(b)$ words of space. To achieve this, we develop a novel technique for performing efficient batched longest common prefix (LCP) queries, using little space. In particular, we show how to answer a batch of $b$ LCP queries using only $O(b)$ words of space, in $O(n\log b)$ time. This technique may be of independent interest, and we expect it to be helpful in other applications in which space usage is a factor. In addition, we show some tradeoffs of construction time and space usage, which are based on time-space tradeoffs of the batched LCP queries. In particular we show that using $O(b\alpha)$ space the construction time is reduced to $O(n\frac{\log^2 b}{\log \alpha} + \frac{\alpha b\log^2 b}{\log \alpha})$. So, for example, if $\alpha = b^\veps$ for a small constant $\veps >0$, then the cost for constructing the sparse suffix tree becomes $O(\frac{1}{\veps}(n\log b + b^{1+\veps}\log b))$, using $O(b^{1+\veps})$ words of space.

\section{Preliminaries}
For a string $T=t_1\cdots t_n$ of size $n$, denote by $T_{i} = t_i\cdots t_n$ the $i^{\text{th}}$ suffix of $T$. The LCP of two suffixes $T_i$ and $T_j$ is denoted by $LCP(T_i,T_j)$, but we will slightly abuse notation and write $LCP(i,j) = LCP(T_i,T_j)$. We denote by $T_{i,j}$ the substring $t_i\cdots t_j$.

We assume the reader is familiar with the suffix tree data structure. For any node $u$ in a (sparse) suffix tree, let $length(u)$ denote the length of the substring corresponding path from the root of the suffix tree to $u$.

\paragraph{Fingerprinting}
We make use of the fingerprinting techniques of Rabin and Karp from~\cite{KR87}. We assume that $T$ is over the integer alphabet $\Sigma = \{1,2,\cdots \sigma \}$, as this will be needed for the fingerprinting. If this is not the case, then we can use perfect hashing (For example a $2$--wise independent hash function into the integers bounded by $\sigma^c$ for some constant $c$ will suffice) for the purpose of the fingerprinting, which works with high probability. This suffices as for fingerprinting purposes we only care if strings are equal, and not about their lexicographical order.

Let $p$ be a prime between $2$ and $n^2$. A fingerprint for a substring $T_{i,j}$, denoted by $\textit{FP}[i,j]$, is the number $\sum_{k=1}^{j} \sigma^{j-k}\cdot t_k \mod p$. Two equal substrings will always have the same fingerprint, however the converse is not true. Luckily, it can be shown that the probability of any two different substrings having the same fingerprint is at most by $n^{-O(1)}$~\cite{KR87}. The exponent in the polynomial can be amplified by a standard constant number of repetitions.

We utilize two important properties of fingerprints. The first is that $\textit{FP}[i,j+1]$ can be computed from $\textit{FP}[i,j]$ in constant time. This is done by the formula $\textit{FP}[i,j+1] = \textit{FP}[i,j]\cdot \sigma + t_{j+1} \mod p$. The second is that the fingerprint of $T_{k,j}$ can be computed in $O(1)$ time from the fingerprint of $T_{i,j}$ and $T_{i,k}$, for $i\leq k \leq j$. This is done by the formula $\textit{FP}[k,j] = \textit{FP}[i,j] - \textit{FP}[i,k]\cdot \sigma^{j-k} \mod p$. Notice however that in order to perform this computation, we must have stored $\sigma^{j-k} \mod p$ as computing it on the fly may be costly.

Our algorithm will be using fingerprinting, and therefore will be correct with high probability. Being that the running time is polynomial in $n$, is is possible to guarantee that the algorithm works with probability at least $1-n^{-O(1)}$, via repeating the fingerprints enough times (but still constant), and the union bound.

\section{Batch LCP Queries}\label{section:batched_lcp}

\subsection{The Algorithm}
Given a string $T$ of size $n$ and a list of $b$ pairs of indices $P$, we wish to compute $LCP(i,j)$ for all $(i,j)\in P$. To do this we perform $\log b$ rounds of computation, where at the $k^{\text{th}}$ round the input is a set of $b$ pairs denoted by $P_k$, where we are guaranteed that for any $(i,j)\in P_k, LCP(i,j)\leq 2^{\log n - (k-1)}$. The goal of the $k^{\text{th}}$ iteration is to decide for any $(i,j)\in P_k$ if $LCP(i,j)\leq 2^{\log n -k}$ or not. In addition, the $k^{\text{th}}$ round will prepare $P_{k+1}$, which is the input for the $(k+1)^{\text{th}}$ round. To begin the execution of the procedure we set $P_0 = P$, as we are always guaranteed that for any $(i,j)\in P$, $LCP(i,j) \leq n = 2^{\log n}$. We will first provide a description of what happens during each of the $\log b$ rounds, and after we will explain how the algorithm uses $P_{\log b}$ to derive $LCP(i,j)$ for all $(i,j)\in P$.

\paragraph{A Single Round}
The $k^{\text{th}}$ round, for $1\leq k \leq \log b$, is executed as follows. We begin by constructing the set $L=\bigcup_{(i,j)\in P_k} \{i-1,j-1,i+2^{\log n -k}, j+2^{\log n -k}\}$ of size $4b$, and construct a perfect hash table for the values in $L$, using a $2$-wise independent hash function into a world of size $b^c$ for some constant $c$ (which with high probability guarantees that there are no collisions). Notice if two elements in $L$ have the same value, then we store them in a list at their hashed value. In addition, for every value in $L$ we store which index created it, so for example, for $i-1$ and $i+2^{\log n -k}$ we remember that they were created from $i$.

Next, we scan $T$ from $t_1$ till $t_n$. When we reach $t_\ell$ we compute $\textit{FP}[1,\ell]$ in constant time from $\textit{FP}[1,{\ell-1}]$. In addition, if $\ell \in L$ then we store $\textit{FP}[1,\ell]$ together with $\ell$ in the hash table. Once the scan of $T$ is completed, for every $(i,j)\in P_k$ we compute $\textit{FP}[i,{i+2^{\log n -k}}]$ in constant time, as we stored $\textit{FP}[1,{i-1}]$ and $\textit{FP}[1,{i+2^{\log n -k}}]$. Similarly we compute $\textit{FP}[j,{j+2^{\log n -k}}]$. Notice that to do this we need to compute $\sigma{2^{\log n -k}} \mod p= \sigma^{\frac{n}{2^k}}$ in $O(\log n -k)$ time which can be easily afforded within our bounds, as one computation suffices for all pairs.

If $\textit{FP}[i,{i+2^{\log n -k}}]\neq \textit{FP}[j,{j+2^{\log n -k}}]$ then it must be that $LCP(i,j) < 2^{\log n -k}$, and so we add $(i,j)$ to $P_{k+1}$. Otherwise, with high probability $LCP(i,j) \geq 2^{\log n -k}$ and so we add $(i+2^{\log n +k},j+2^{\log n +k})$ to $P_{k+1}$. Notice there is a natural bijection between pairs in $P_{k-1}$ and pairs in $P$ following from the method of constructing the pairs for the next round. For each pair in $P_{k+1}$ we will remember which pair in $P$ originated it, which can be easily transferred when $P_{k+1}$ is constructed from $P_k$.

\paragraph{LCP on Small Strings}
After the $\log b$ rounds have taken place, we know that for every $(i,j)\in P_{\log b}$, $LCP(i,j)\leq 2^{\log n - \log b} = \frac nb$. For each such pair, we spend $O(\frac nb)$ time in order to exactly compute $LCP(i,j)$. Notice that this is performed for $b$ pairs, so the total cost is $O(n)$ for this last phase. We then construct $P_{\textit{final}} = \{(i+LCP(i,j), j+LCP(i,j)) : (i,j)\in P_{\log b}\}$. For each $(i,j)\in P_{\textit{final}}$ denote by $(i_0,j_0)\in P$ the pair which originated $(i,j)$. We claim that for any $(i,j)\in P_{\textit{final}}$, $LCP(i_0,j_0) = i - i_0$.

\subsection{Runtime and Correctness}
Each round takes $O(n)$ time, and the number of rounds is $O(\log b)$ for a total of $O(n\log b)$ for all rounds. In addition, the work executed for computing $P_{\textit{final}}$ is an additional $O(n)$.

The following lemma on LCPs will be helpful in proving the correctness of the batched LCP query.
\begin{lemma}\label{lemma:LCP_divided}
For any $1\leq i,j \leq n$, for any $0\leq m \leq LCP(i,j)$, it holds that $LCP(i+m,j+m)+m = LCP(i,j)$.
\end{lemma}
\begin{proof}
This follows directly from the definition of LCP.
\end{proof}

We now proceed on to prove that for any $(i,j)\in P_{\textit{final}}$, $LCP(i_0,j_0) = i - i_0$. Lemma~\ref{lemma:bound_per_round} shows that the algorithm behaves as expected during the $\log b$ rounds, and Lemma~\ref{lemma:final_round} proves that the work done in the final round suffices for computing the LCPs.

\begin{lemma}\label{lemma:bound_per_round}
At round $k$, for any $(i_k,j_k)\in P_k$, $i_k-i_0 \leq LCP(i_0,j_0)\leq i_k-i_0+2^{\log n -k}$, assuming the fingerprints do not give a false positive.
\end{lemma}
\begin{proof}
The proof is by induction on $k$. For the base, $k=0$ and so $P_0 = P$ meaning that $i_k=i_0$. Therefore, $i_k-i_0 = 0 \leq LCP(i_0,j_0) \leq 2^{\log n} = n$, which is always true. For the step, we assume correctness for $k-1$ and we prove for $k$ as follows. By the induction hypothesis, for any $(i_{k-1},j_{k-1})\in P_{k-1}$, $i-i_0 \leq LCP(i_0,j_0)\leq i-i_0+2^{\log n -k+1}$. Let $(i_k,j_k)$ be the pair in $P_k$ corresponding to $(i_{k-1}, j_{k-1})$ in $P_{k-1}$. If $i_k = i_{k-1}$ then $LCP(i_{k-1},j_{k-1}) < 2^{\log n -k}$. Therefore,
\begin{align*}
i_k - i_0
& = i_{k-1}-i_0 \\
& \leq LCP(i_0,j_0) \\
& \leq i_{k-1} - i_0 + LCP(i_{k-1},j_{k-1}) \\
& \leq i_{k} - i_0 + 2^{\log n -k}.
\end{align*}

If $i_k = i_{k-1} + 2^{\log n -k}$ then $\textit{FP}[i,{i+2^{\log n -k}}] = \textit{FP}[j,{j+2^{\log n -k}}]$, and being as we assume that the fingerprints do not give produce false positives, $LCP(i_{k-1},j_{k-1} ) \geq 2^{\log n -k}$. Therefore,
\begin{align*}
i_k - i_0
& = i_{k-1} + 2^{\log n -k}-i_0 \\
& \leq i_{k-1}-i_0 + LCP(i_{k-1},j_{k-1}) \\
&= LCP(i_0,j_0) \\
& \leq i_{k-1} - i_0  + 2^{\log n -k+1}\\
& = i_k - i_0+ 2^{\log n -k},
\end{align*}
where the third equality holds from Lemma~\ref{lemma:LCP_divided}, and the fourth inequality holds as $LCP(i_0,j_0) =  i_{k-1}-i_0 + LCP(i_{k-1},j_{k-1})$ (which is the third equality), and $LCP(i_{k-1},j_{k-1})  \leq 2^{\log n -k+1}$ by the induction hypothesis.
\end{proof}

\begin{lemma}\label{lemma:final_round}
For any $(i,j)\in P_{\textit{final}}$, $LCP(i_0,j_0) = i-i_0 (=j-j_0)$.
\end{lemma}
\begin{proof}
Using Lemma~\ref{lemma:bound_per_round} with $k=\log b$ we have that for any $(i_{\log b},j_{\log b})\in P_{\log b}$, $i_{\log b} - i_0 \leq LCP(i_0,j_0) \leq i_{\log b} - i_0 +2^{\log n - \log b} = i_{\log b} - i_0 + \frac nb$. Being that $LCP(i_{\log b},j_{\log b}) \leq 2^{\log n - \log b}$ it must be that $LCP(i_0,j_0)  = i_{\log b} - i_0 + LCP(i_{\log b},j_{\log b}) $. Notice that $i_{\textit{final}} = i_{\log b} + LCP(i_{\log b},j_{\log b})$. Therefore, $LCP(i_0,j_0) = i_{\textit{final}} - i_0$ as required.
\end{proof}

Notice that the space used in each round is the set of pairs and the hash table for $L$, both of which require only $O(b)$ words of space. Thus, we have obtained the following.

\begin{theorem}\label{thm:batched_lcp}
It is possible to compute the $LCP$ of $b$ pairs of suffixes of a string $T$ of size $n$ in $O(n\log b)$ time using $O(b)$ space.
\end{theorem}

We discuss several other time/space tradeoffs in Section~\ref{section:tradeoff}

\section{Constructing the Sparse Suffix Tree}\label{sec:SST_construction}
The procedure for constructing the sparse suffix tree using only $O(b)$ space is split into two stages. In the first stage, we lexicographically sort the $b$ suffixes. In the second stage, we compute the $LCP$ of every two consecutive suffixes in the ordered list, and use those LCPs to simulate a DFS traversal on the sparse suffix tree, constructing the sparse suffix tree as we go along.

\subsection{Stage 1: Suffix Sorting}
We can use batched LCP queries in order to compare $b$ pairs of suffixes, as once the LCP of two suffixes is known, deciding which of the two is lexicographically smaller than the other takes constant time by examining the first two characters that differ in said suffixes. So we are interested in performing roughly $O(\log b)$ sets of $b$ comparisons each in order to sort the suffixes, where each set of comparisons is performed via batched LCP queries. One way to do this is to simulate a sorting network on the $b$ suffixes of depth $\log b$~\cite{AKS83}. Unfortunately, such known networks have very large constants hidden in them, and are generally considered impractical~\cite{Paterson90}. There are some practical networks with depth $\log ^2 b$ such as~\cite{Batcher68}, however, we wish to do better.

What we chose to do is simulate the quick-sort algorithm by each time picking a random suffix called the pivot, and lexicographically comparing all of other $b-1$ suffixes to the pivot. Once a partition is made to the set of suffixes which are lexicographically smaller than the pivot, and the set of suffixes which are lexicographically larger than the pivot, we recursively sort each set in the partition with the following modification. Each level of the recursion tree is performed concurrently using one batched LCP query for the entire level. The number of comparisons performed in each level is always bounded by $O(b)$, so we may use Theorem~\ref{thm:batched_lcp}. Furthermore, with high probability, the number of levels in the randomized quicksort is $O(\log b)$. Thus the total amount of time spent, with high probability is $O(n\log^2 b)$. Notice that from a theoretical point of view, it is possible to have a deterministic runtime of the same magnitude using sorting networks.

Notice that once the suffixes have been sorted, then we have in fact computed the sparse suffix array for the $b$ suffixes. Hence we have obtained the following.

\begin{theorem}\label{thm:sparse_suffix_array}
There exists a randomized algorithm that with high probability constructs the sparse suffix array for a string $T$ of size $n$ and a set of any $b$ indices in $T$ in $O(n \log^2 b)$ time in the worst case.
\end{theorem}

\subsection{Stage 2: Traversing the Sparse Suffix Tree}
Let $S=\{T_{i_1},\cdots T_{i_b}\}$ be the ordered list of suffixes for which we wish to construct the sparse suffix tree. Then we begin by computing $LCP(i_j,i_{j+1})$ for all $1\leq j \leq b-1$. This takes $O(n\log b)$ time using Theorem~~\ref{thm:batched_lcp}. Now we wish to simulate a DFS traversal on the sparse suffix tree in order to construct it. This is done as follows.

The algorithm begins by creating a node which will be the root of the sparse suffix tree, and denoted by $r$. Denote by $Q_j\subset S$ the set of first $j-1$ suffixes in $S$, taken by lexicographical order. We will iteratively construct the sparse suffix tree for $Q_j$ for each $1\leq j\leq b$. Denote the sparse suffix tree for $Q_j$ by $ST_j$. For $j=1$, $ST_1$ is simply $r$ with one child that is the single node for $T_{i_1}$. Assume we have $ST_{j-1}$; we show how to use it to construct $ST_j$. We need to locate the location of the node $u$ which will be the lowest common ancestor of the leaf corresponding to $T_{i_{j-1}}$ and the leaf corresponding to $T_{i_j}$. To do this we traverse the path in $ST_{j-1}$ from the leaf corresponding to $T_{i_{j-1}}$ to $r$, and each time we reach a node $v$ on this path, we compare the length of its label $length(v)$ to $LCP(i_{j-1},i_j)$. If the two are equal, then this is the node $u$ we are searching for, and we insert $T_{i_j}$ as a child of this node. If $length(v) > LCP(i_{j-1},i_j)$ then we need to continue up the path. If $length(v) < LCP(i_{j-1},i_j)$ then the node $u$ needs to be inserted as a child of $u$, breaking the edge going from $v$ towards the leaf corresponding to $T_{i_{j-1}}$, which is the node we previously encountered while traversing the path. When $u$ is inserted, we set $length(u) \leftarrow LCP({i_{j-1}},{i_j})$, and add the leaf corresponding to $T_{i_j}$ as a child of $u$. Notice that the $label(r)=0$ so this process will in the worst case end at $r$, with $u=r$. 

This process simulates a DFS search on the sparse suffix tree, and so the total time cost for this DFS is $O(b)$. Thus we have obtained the following.

\begin{theorem}\label{thm:sparse_suffix_tree}
There exists a randomized algorithm that with high probability constructs the sparse suffix tree for a string $T$ of size $n$ and a set of any $b$ indices in $T$ in $O(n \log^2 b)$ time in the worst case.
\end{theorem}

\section{Time-Space Tradeoffs for Batched LCP Queries}\label{section:tradeoff}
We provide an overview of the technique used to obtain the time-space tradeoff for the batched LCP process, as it closely follow those of Section~\ref{section:batched_lcp}. In Section~\ref{section:batched_lcp} the algorithm simulates concurrent binary searches in order to determine the  $LCP$ of each input pair (with some extra work at the end). The idea for obtaining the tradeoff is to generalize the binary search to an $\alpha$-ary search. So in the $k^{\text{th}}$ round the input is a set of $b$ pairs denoted by $P_k$, where we are guaranteed that for any $(i,j)\in P_k, LCP(i,j)\leq 2^{\log n - (k-1)\log \alpha}$, and the goal of the $k^{\text{th}}$ iteration is to decide for any $(i,j)\in P_k$ if $LCP(i,j)\leq 2^{\log n -k\log \alpha}$ or not. From a space perspective, this means that we need $O(\alpha b)$ space in order to compute $\alpha$ fingerprints per each index in any $(i,j)\in P_k$. From a time perspective, we only need to perform $O(\log_{\alpha} b)$ rounds before we may begin the final round. However, each round now costs $O(n+\alpha b)$. So the total cost for a batched LCP query is $O(\log_{\alpha} b (n+\alpha b)) = O(n\frac{\log b}{\log \alpha} + \frac{\alpha b\log b}{\log \alpha})$, and the total time cost for constructing the sparse suffix tree is $O(n\frac{\log^2 b}{\log \alpha} + \frac{\alpha b\log^2 b}{\log \alpha})$.

If, for example, $\alpha = b^\veps$ for a small constant $\veps >0$, then the cost for constructing the sparse suffix tree becomes $O(\frac{1}{\veps}(n\log b + b^{1+\veps}\log b))$, using $O(b^{1+\veps})$ words of space.

{\small
\bibliographystyle{alphainit}
\bibliography{tsvi}

\begin{thebibliography}{KKS11}

\bibitem[AKS83]{AKS83}
M.~Ajtai, J.~Koml{\'o}s, and E.~Szemer{\'e}di.
\newblock An o(n log n) sorting network.
\newblock In {\em Proceedings of the 15th Annual ACM Symposium on Theory of
  Computing}, pages 1--9, 1983.

\bibitem[ALS99]{ALS99}
A.~Andersson, N.~J. Larsson, and K.~Swanson.
\newblock Suffix trees on words.
\newblock {\em Algorithmica}, 23(3):246--260, 1999.

\bibitem[Bat68]{Batcher68}
K.~E. Batcher.
\newblock Sorting networks and their applications.
\newblock In {\em AFIPS Spring Joint Computing Conference}, pages 307--314,
  1968.

\bibitem[FF07]{FF07}
P.~Ferragina and J.~Fischer.
\newblock Suffix arrays on words.
\newblock In {\em CPM}, pages 328--339, 2007.

\bibitem[KKS11]{KKS11}
R.~Kolpakov, G.~Kucherov, and T.~A. Starikovskaya.
\newblock Pattern matching on sparse suffix trees.
\newblock In {\em First International Conference on Data Compression,
  Communications and Processing}, pages 92--97, 2011.

\bibitem[KR87]{KR87}
R.~M. Karp and M.~O. Rabin.
\newblock Efficient randomized pattern-matching algorithms.
\newblock {\em IBM Journal of Research and Development}, 31(2):249--260, 1987.

\bibitem[KU96]{KU96}
J.~K{\"a}rkk{\"a}inen and E.~Ukkonen.
\newblock Sparse suffix trees.
\newblock In {\em Computing and Combinatorics, Second Annual International
  Conference}, pages 219--230, 1996.

\bibitem[Pat90]{Paterson90}
M.~Paterson.
\newblock Improved sorting networks with o(log n) depth.
\newblock {\em Algorithmica}, 5(1):65--92, 1990.

\end{thebibliography}
}

\end{document}